\documentclass[12pt]{article}
\usepackage{dsfont}
\usepackage{graphicx}
\usepackage{amsmath}
\usepackage{amssymb}
\newtheorem{theorem}{Theorem}[section]

\newtheorem{definition}{Definition}[section]

\usepackage{arydshln}
\newenvironment{proof}[1][Proof]{\noindent\textbf{#1.} }{\ }
\begin{document}
\title{Local optimality of a coherent feedback scheme for distributed entanglement generation: the idealized infinite bandwidth limit}
\author{Zhan Shi and Hendra I. Nurdin % <-this % stops a space
%\thanks{This work was not supported by any organization}% 
\thanks{
Z. Shi and H. Nurdin are with School of Electrical Engineering and 
Telecommunications,  The University of New South Wales,  
Sydney NSW 2052, Australia (e-mail: h.nurdin@unsw.edu.au, zhan.shi@student.unsw.edu.au).} 
}
\maketitle

%\thispagestyle{empty} \pagestyle{empty}

%%%%%%%%%%%%%%%%%%%%%%%%%%%%%%%%%%%%%%%%%%%%%%%%%%%%%%%%%%%%%%%%%%%%%%%%%%%%%%%%%%%%%%%%%%%%%%%%%%%%%%%%%%%%%%
%%%%%%%%%%%%%%%%%%%%%%%%%%%%%%%%%%%%%%%%%%%%%%%%%%%%%%%%%%%%%%%%%%%%%%%%%%%%%%%%%%%%%%%%%%%%%%%%%%%%%%%%%%%%%%
%%%%%%%%%%%%%%%%%%%%%%%%%%
%%%%%%%%%%%%%%%%%%%%%%%%%%%%%%%%%%%%%%%%%%%%%%%%%%%%%%%%%%%%%%%%%%%%%%%%%%%%%%%%%%%%%%%%%%%%%%%%%%%%%%%%%%%%%%
%%%%%%%%%%%%%%%%%%%%%%%%%%%%%%%%%%%%%%%%%%%%%%%%%%%%%%%%%%%%%%%%%%%%%%%%%%%%%%%%%%%%%%%%%%%%%%%%%%%%%%%%%%%%%%
\begin{abstract}
The purpose of this paper is to prove a local optimality property of a recently proposed coherent feedback configuration for distributed generation of EPR entanglement using two nondegenerate optical parametric amplifiers (NOPAs) in the idealized infinite bandwidth limit. This local optimality is with respect to a class of similar coherent feedback configurations but employing different unitary scattering matrices, representing different scattering of propagating signals within the network. The infinite bandwidth limit is considered as it significantly simplifies the analysis, allowing local optimality criteria to be explicitly verified. Nonetheless, this limit is relevant for the finite bandwidth scenario as it provides an accurate approximation to the EPR entanglement in the low frequency region where EPR entanglement exists.
\end{abstract}

%%%%%%%%%%%%%%%%%%%%%%%%%%%%%%%%%%%%%%%%%%%%%%%%%%%%%%%%%%%%%%%%%%%%%%%%%%%%%%%%%%%%%%%%%%%%%%%%%%%%%%%%%%%%%%
%%%%%%%%%%%%%%%%%%%%%%%%%%%%%%%%%%%%%%%%%%%%%%%%%%%%%%%%%%%%%%%%%%%%%%%%%%%%%%%%%%%%%%%%%%%%%%%%%%%%%%%%%%%%%%
%%%%%%%%%%%%%%%%%%%%%%%%%%%%%%%%%%%%%%%%%%%%%%%%%%%%%%%%%%%%%%%%%%%%%%%%%%%%%%%%%%%%%%%%%%%%%%%%%%%%%%%%%%%%%%
\section{Introduction}
\label{sec:intro}
Entanglement is a quantum phenomenon in which states (represented by density operators) of a composite system composed of several quantum subsystems cannot be written as a convex combination of tensor products of the states of the subsystems. 
Such entangled states have, in recent decades, been of much interest as a resource for quantum information applications, such as for quantum communication \cite{Bowen2004,Weedbrook2012}. 
In particular, Einstein-Podolski-Rosen (EPR)-like entanglement,  generated in the continuous variables such as the amplitude and phase quadratures of a Gaussian optical field, has evoked considerable interest over discrete-variable entanglement, such as entanglement in finite-level systems like qubits, because EPR entangled pairs can be prepared easily and rapidly in quantum optics. 
In this paper, we are interested in EPR entanglement between two propagating continuous-mode Gaussian fields. Such a kind of entanglement is more accessible compared to EPR entanglement between a pair of single-mode fields produced in, say, inside an optical cavity \cite{Bowen2004, Braunstein2005}. 

EPR entanglement between continuous-mode Gaussian fields can be realized by two-mode squeezed states produced as the output of a nondegenerate optical parametric amplifier (NOPA). By pumping a strong coherent beam (which can be regarded as an undepleted classical light) to a crystal inside the cavity of the NOPA, two vacuum modes of the cavity interact with the pump beam, and photons escaping the cavity through its partially transmissive mirrors generate two output beams that are squeezed in amplitude and phase quadratures. If the two outgoing fields are squeezed below the quantum shot-noise limit, they are considered as EPR entangled beams \cite{Ou1992, Vitali2006}.  The input/output block representation of a NOPA ($G_i$) is shown as Fig.~\ref{fig:single-NOPA}. The NOPA has four ingoing fields and four outgoing fields. Among the inputs, $\xi_{loss,a,i}$ and $\xi_{loss,b,i}$ are amplification losses, caused by unwanted vacuum modes coupled into the cavity. As the two outputs corresponding to the loss fields  $\xi_{loss,a,i}$ and $\xi_{loss,b,i}$ are not of interest in this work, they are not shown in the figure. Note Fig.~\ref{fig:single-NOPA} only presents the ingoing and outgoing noises of interest, and does not show the pump beam.
%%%%%%%%%%%%%%%%%%%%%%%%%%%%%%%%%%%%%%%%%%%%%%%%%%%%%%%%%%%%%%%%%%%%%%%%%%%%%%%%%%%%%%%%%%%%%%%%%%%%%%%%%%%%%%%%%%%%%%%%%%%%%%%%%%%%%%%%%%%%%%%%%%%%%%%%%%%%%%%%%%%%%%%%%%%%%%%%%%%%%%%%%%%%%%%%%%%%%%%%%%%%%%%%%%%%%%%%%%%%%%%%
\begin{figure}[htbp]
\begin{center}
\includegraphics[scale=0.4]{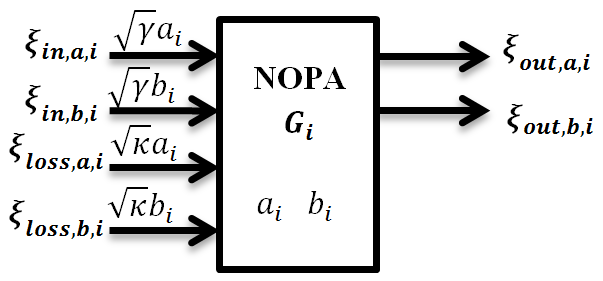}
\caption{Input/output block representation of a NOPA.}\label{fig:single-NOPA}
\end{center}
\end{figure}
%%%%%%%%%%%%%%%%%%%%%%%%%%%%%%%%%%%%%%%%%%%%%%%%%%%%%%%%%%%%%%%%%%%%%%%%%%%%%%%%%%%%%%%%%%%%%%%%%%%%%%%%%%%%%%%%%%%%%%%%%%%%%%%%%%%%%%%%%%%%%%%%%%%%%%%%%%%%%%%%%%%%%%%%%%%%%%%%%%%%%%%%%%%%%%%%%%%%%%%%%%%%%%%%%%%%%%%%%%%%%%%%

In a previous work \cite{SN2015qip}, we have proposed a novel dual-NOPA coherent feedback system to produce EPR entangled propagating Gaussian fields, as shown in  Fig.~\ref{fig:dual-NOPA-cfb}. It was shown that this scheme can produce better EPR entanglement between the propagating Gaussian fields $\xi_{out,a,2}$ and $\xi_{out,b,1}$  (in the sense of producing more two-mode squeezing between quadratures of the fields) for the same amount of total pump power used in the two NOPAs, and displays more tolerance to transmission losses in the system, as compared to a conventional single NOPA and a cascaded two-NOPA system. 

%%%%%%%%%%%%%%%%%%%%%%%%%%%%%%%%%%%%%%%%%%%%%%%%%%%%%%%%%%%%%%%%%%%%%%%%%%%%%%%%%%%%%%%%%%%%%%%%%%%%%%%%%%%%%%%%%%%%%%%%%%%%%%%%%%%%%%%%%%%%%%%%%%%%%%%%%%%%%%%%%%%%%%%%%%%%%%%%%%%%%%%%%%%%%%%%%%%%%%%%%%%%%%%%%%%%%%
\begin{figure}[htbp]
\begin{center}
\includegraphics[scale=0.35]{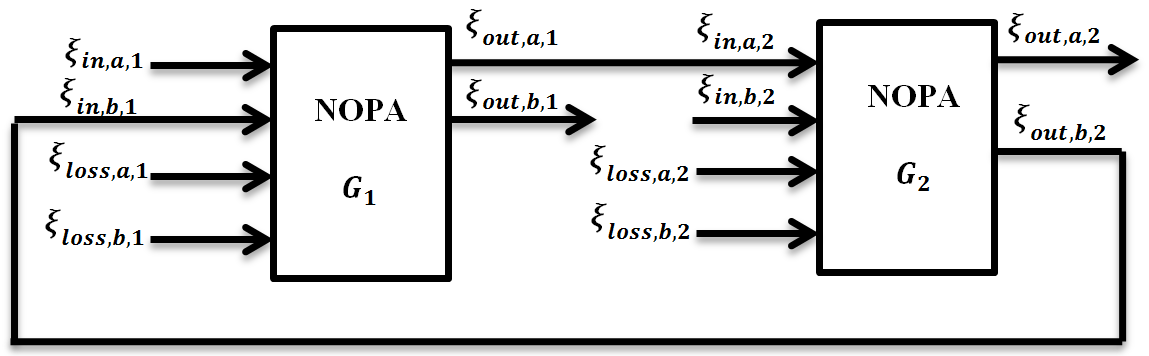}
\caption{The dual-NOPA coherent feedback network.}\label{fig:dual-NOPA-cfb}
\end{center}
\end{figure}
%%%%%%%%%%%%%%%%%%%%%%%%%%%%%%%%%%%%%%%%%%%%%%%%%%%%%%%%%%%%%%%%%%%%%%%%%%%%%%%%%%%%%%%%%%%%%%%%%%%%%%%%%%%%%%%%%%%%%%%%%%%%%%%%%%%%%%%%%%%%%%%%%%%%%%%%%%%%%%%%%%%%%%%%%%%%%%%%%%%%%%%%%%%%%%%%%%%%%%%%%%%%%%%%%%%%%%

In a subsequent work \cite{SN2015acc}, we presented a linear quantum system consisting of two NOPAs connected to a static passive linear network, realizable by a network of beam splitters, mirrors and phase shifters, that are connected in a more general coherent feedback configuration, see Fig.~\ref{fig:system_paper2}.  
Here, the system is ideally lossless, that is, there are no transmission and amplification losses influencing the system. Hence, each NOPA is simplified to have only two ingoing fields, without amplification losses, as shown in Fig.~\ref{fig:system_paper2}.
The transformation implemented by the passive network in this configuration is represented by a $6 \times 6$ complex unitary matrix $\tilde S$. 
By employing a modified steepest descent algorithm, with the matrix corresponding to the dual-NOPA coherent feedback network shown in Fig.~\ref{fig:dual-NOPA-cfb} as a starting point, we optimized the EPR entanglement at frequency $\omega=0$, with respect of the transformation matrix $\tilde S$ of the passive network.

%%%%%%%%%%%%%%%%%%%%%%%%%%%%%%%%%%%%%%%%%%%%%%%%%%%%%%%%%%%%%%%%%%%%%%%%%%%%%%%%%%%%%%%%%%%%%%%%%%%%%%%%%%%%%%%%%%%%%%%%%%%%%%%%%%%%%%%%%%%%%%%%%%%%%%%%%%%%%%%%%%%%%%%%%%%%%%%%%%%%%%%%%%%%%%%%%%%%%%%%%%%%%%%%%%%%%%
\begin{figure}[htbp]
\begin{center}
\includegraphics[scale=0.35]{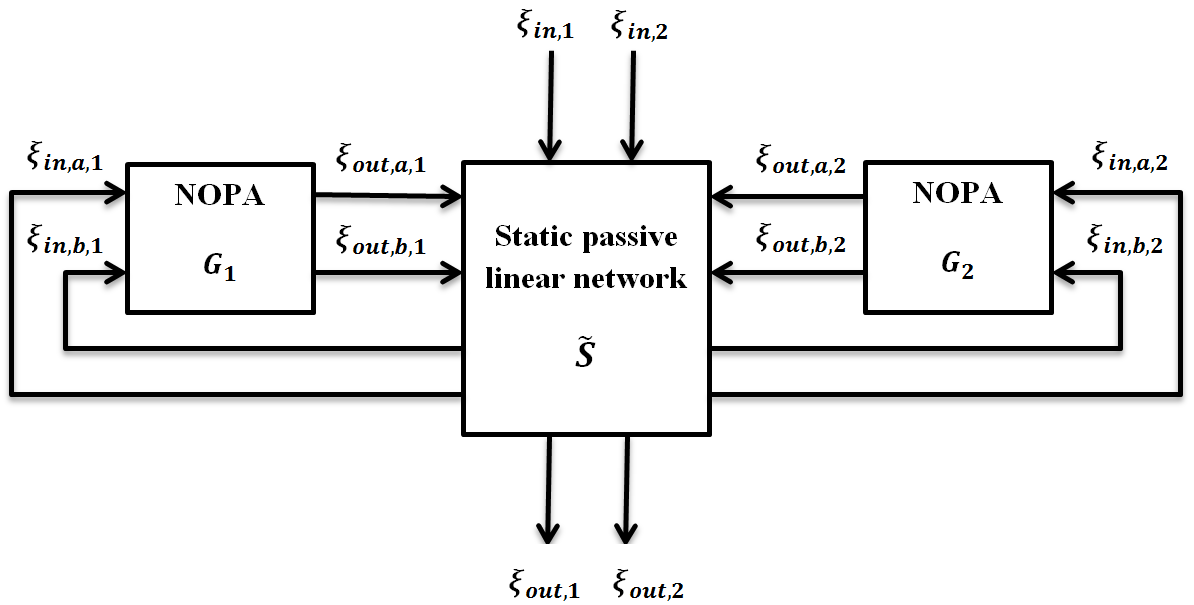}
\caption{A coherent-feedback system consisting of two NOPAs and a static passive network with six imputs and six outputs from \cite{SN2015acc}.}\label{fig:system_paper2}
\end{center}
\end{figure}
%%%%%%%%%%%%%%%%%%%%%%%%%%%%%%%%%%%%%%%%%%%%%%%%%%%%%%%%%%%%%%%%%%%%%%%%%%%%%%%%%%%%%%%%%%%%%%%%%%%%%%%%%%%%%%%%%%%%%%%%%%%%%%%%%%%%%%%%%%%%%%%%%%%%%%%%%%%%

In this paper, we employ the steepest descent method to optimize a coherent feedback system shown in Fig.~\ref{fig:system}. The system contains two NOPAs and a static passive linear network described by a $2 \times 2$ complex unitary matrix $\tilde S$.  This system is a more restricted class of configuration than the one as shown in Fig.~\ref{fig:system_paper2}; it can be seen that the configuration in Fig.~\ref{fig:system} is a special case of the configuration in Fig.~\ref{fig:system_paper2}. 
Moreover, different from our previous work in \cite{SN2015acc}, in which the system shown in Fig.~\ref{fig:system_paper2} is considered  lossless, here we take the effect of transmission losses along channels and amplification losses of NOPAs into account. However, we neglect time delays in transmission. The effect of delays on EPR entanglement generated from related systems can be found  in our previous works \cite{SN2015qip, SN2015qic}.
In addition, unlike the work in \cite{SN2015acc}, the system is considered ideally static, that is, we consider the limit where the NOPAs are approximated as static devices with an infinite bandwidth \cite{Gough2010}.  The merits of studying this infinite bandwidth limit are twofold: (i) it allows a simplified analysis of the system, and (ii) calculations in the infinite bandwidth setting gives a very good approximation to the EPR entanglement in the low frequency region, discussed further in Section~\ref{sec:NOPA}. In this infinite bandwidth setting, we show explicitly that the choice of the scattering matrix in the scheme of \cite{SN2015qip} is in a certain sense locally optimal with respect to all possible choices of scattering matrices $\tilde S$ in the coherent feedback configuration of Fig.~\ref{fig:system}, under certain values of the effective amplitude of the pump laser driving the NOPA.  Note that there may exist another scattering matrix as a local minimizer that yields better EPR entanglement than the network shown in Fig.~\ref{fig:dual-NOPA-cfb}. Searching for such a scattering matrix can be a topic for future research.
%%%%%%%%%%%%%%%%%%%%%%%%%%%%%%%%%%%%%%%%%%%%%%%%%%%%%%%%%%%%%%%%%%%%%%%%%%%%%%%%%%%%%%%%%%%%%%%%%%%%%%%%%%%%%%%%%%%%%%%%%%%%%%%%%%%%%%%%%%%%%%%%%%%%%%%%%%%%%%%%%%%%%%%%%%%%%%%%%%%%%%%%%%%%%%%%%%%%%%%%%%%%%%%%%%%%%%%%%%%%%%%%
\begin{figure*}[htbp]
\begin{center}
\includegraphics[scale=0.5]{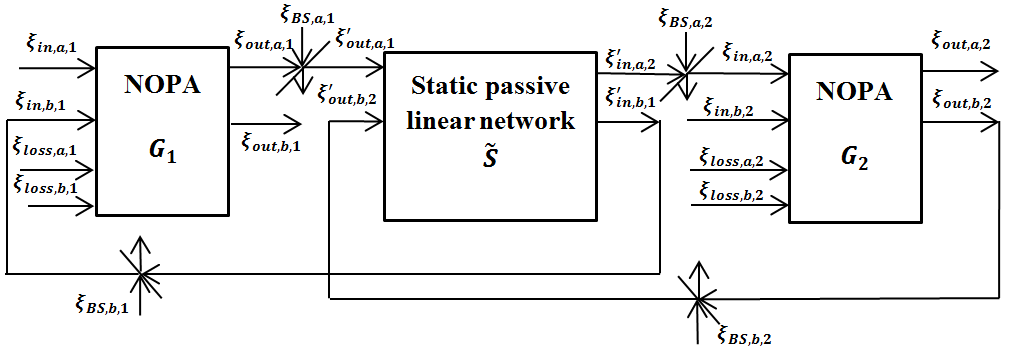}
\caption{A coherent-feedback system consisting of two NOPAs and a static passive network with two inputs and two outputs.}\label{fig:system}
\end{center}
\end{figure*}
%%%%%%%%%%%%%%%%%%%%%%%%%%%%%%%%%%%%%%%%%%%%%%%%%%%%%%%%%%%%%%%%%%%%%%%%%%%%%%%%%%%%%%%%%%%%%%%%%%%%%%%%%%%%%%%%%%%%%%%%%%%%%%%%%%%%%%%%%%%%%%%%%%%%%%%%%%%%%%%%%%%%%%%%%%%%%%%%%%%%%%%%%%%%%%%%%%%%%%%%%%%%%%%%%%%%%%%%%%%%%%%%

The structure of the rest of this paper is as follows. We begin in Section \ref{sec:prelim} by giving a brief review of linear quantum systems,  EPR entanglement between two continuous-mode fields, and linear transformations implemented by a NOPA in the infinite bandwidth limit.  Section \ref{sec:system-model} describes the system of interest. In Section \ref{sec:optimization}, we discuss the optimization of the system. Finally, we draw a short conclusion in Section \ref{sec:conclusion}. 

%%%%%%%%%%%%%%%%%%%%%%%%%%%%%%%%%%%%%%%%%%%%%%%%%%%%%%%%%%%%%%%%%%%%%%%%%%%%%%%%  
%%%%%%%%%%%%%%%%%%%%%%%%%%%%%%%%%%%%%%%%%%%%%%%%%%%%%%%%%%%%%%%%%%%%%%%%%%%%%%%%
%%%%%%%%%%%%%%%%%%%%%%%%%%%%%%%%%%%%%%%%%%%%%%%%%%%%%%%%%%%%%%%%%%%%%%%%%%%%%%%%
%%%%%%%%%%%%%%%%%%%%%%%%%%%%%%%%%%%%%%%%%%%%%%%%%%%%%%%%%%%%%%%%%%%%%%%%%%%%%%%%
\section{Preliminaries}
\label{sec:prelim}
The notations used in this paper are as follows: $\imath=\sqrt{-1}$ and $\operatorname{Re}$ denotes the real part of a complex quantity. The conjugate of a matrix is denoted by $\cdot^\#$, $\cdot^T$ denotes the transpose of a matrix of numbers or operators and $\cdot^*$ denotes (i) the complex conjugate of a number, (ii) the conjugate transpose of a matrix, as well as (iii) the adjoint of an operator. $O_{m\times n}$ is an $m$ by $n$ zero matrix (if $m=n$ then we simply write $O_m$), and $I_n$ is an $n$ by $n$ identity matrix.  Trace operator is denoted by  $\operatorname{Tr[\cdot]}$ and tensor product is $\otimes$.  $\delta(t)$ denotes the Dirac delta function.
\subsection{Linear quantum systems}
\label{sec:linear_sys}
Here we consider an open linear quantum system without a scattering process. The linear system contains $n$-bosonic modes $a_j(t)~(j=1,\ldots, n)$ satisfying the commutation relations $[a_i(t), a_j(t)^*]=\delta_{ij}$, $m$-incoming boson fields $\xi_{in,i}(t)~(i=1,\ldots, m)$ in the vacuum state, which obey the commutation relations $[\xi_{in,j}(t), \xi_{in,j}(s)^*]=\delta(t-s)$, as well as two outgoing fields $\xi_{out,k}(t)~(k=1,2)$ which are Gaussian continuous-mode fields. A continuous-mode field means that the field contains a continuum of modes in a continuous range of frequencies.  Note that, a system may have more than two outputs. 
However, as we are only interested in the entanglement generated by a certain pair of outgoing fields, in this work we will only be interested in a particular pair of output fields, labelled ${out,1}$ and ${out,2}$ in the following.
The time-varying interaction Hamiltonian between the system and its environment is $H_{\rm int}(t) = \imath (\xi(t)^*L - L^* \xi(t))$, where $\xi(t)=[\xi_{in,1}(t),\ldots \xi_{in,m}(t)]^T $, $L=[L_1,\ldots, L_l]^T$ and $L_j  (j=1, 2, \cdots, l)$ is the $j$-th system coupling operator. In the Heisenberg picture, time evolutions of a mode $a_j$ and an outgoing field operator $\xi_{out,i}$ are \cite{bGardiner2004,Nurdin2009}:
\begin{align}
a_j(t)=&U(t)^* a_j U(t),\nonumber \\
\xi_{out,i}(t)=&U(t)^*\xi_{in,i}(t)U(t),
\end{align}
where $U(t)={\rm exp}^{\hspace{-0.5cm}\longrightarrow}~(-i\int_0^t H_{\rm int}(s)ds)$ is a unitary process obeying the quantum white noise Schr\"{o}dinger equation $\dot{U}(t)=-\imath H_{\rm int}(t)U(t)$. 
However, this is not an ordinary Schr\"{o}dinger equation as the interaction Hamiltonian $H_{int}(t)$  is a time-varying observable involving the singular quantum white noise processes $\xi(t)$.  This quantum white noise equation has to be interpreted correctly within the framework of quantum stochastic calculus, for details see \cite{bGardiner2004, inbBelavkin2008, bWiseman2010, Gough2003}.
Employing quantum stochastic calculus,  dynamics of a linear quantum system is described by quantum Langevin equations and can be written in the following form
\begin{eqnarray}
    \dot{z}(t)&=&Az(t)+B\xi(t), \label{eq:dynamics} \\
     \xi_{out}(t)&=&Cz(t)+D\xi(t). \label{eq:output}
\end{eqnarray}
where
\begin{eqnarray}
    z&=&(a_1^q, a_1^p, \ldots, a_n^q, a_n^p)^T, \nonumber\\
    \xi &=&(\xi_{1}^q, \xi_{1}^p, \ldots, \xi_{m}^q, \xi_{m}^p)^T, \nonumber\\
    \xi_{out}&=&(\xi_{out,1}^q, \xi_{out,1}^p, \xi_{out,2}^q, \xi_{out,2}^p)^T, \label{eq:vector}
\end{eqnarray}
with {\it quadratures} \cite{inbBelavkin2008,bWiseman2010}
\begin{eqnarray}
a_j^q &=& a_j+a_j^*, \quad a_j^p = (a_j-a_j^*)/i, \nonumber \\
\xi_j^q &=& \xi_j+\xi_j^*, \quad \xi_j^p = (\xi_j-\xi_j^*)/i. \label{eq: quadratures}
\end{eqnarray}
The linear model described above is ubiquitous in fields such as quantum optics, optomechanics, and superconducting circuits, and are employed to describe the equations of motion for devices as diverse as optical cavities, optical parametric amplifiers, optical cavities with moving mirrors, cold atomic ensembles, and transmission line resonators, under appropriate assumptions on the system's parameters. 
  
\subsection{EPR entanglement between two continuous-mode fields}
\label{sec:entanglement}
We keep in mind here that in this paper, we investigate EPR entanglement between two continuous-mode Gaussian fields rather than entanglement between two single-mode Gaussian fields. In the latter case, the degree of entanglement can be assessed via the logarithmic negativity as an entanglement measure, see, e.g., \cite{Laurat2005}. However, this measure is not directly applicable  to continuous-mode fields. Instead, the EPR entanglement of two freely propagating fields containing a continuum of modes, say $\xi_{out,1}$ and $\xi_{out,2}$, can be evaluated in the frequency domain by the two-mode squeezing spectra $V_+(\imath\omega)$  and $V_-(\imath\omega)$ \cite{Braunstein2005,Ou1992,Vitali2006}, that will be defined below.

The Fourier transform of $f(t)$ is defined as $F\left(\imath\omega\right)=\frac{1}{\sqrt{2\pi}}\int_{-\infty}^{\infty} f\left(t\right)e^{-\imath\omega t} dt$. Similarly, we have the Fourier transforms of $\xi_{out,1}(t)$, $\xi_{out,2}(t)$, $z(t)$ and $\xi(t)$ in (\ref{eq:dynamics}) and (\ref{eq:output}) as  $\tilde \Xi_{out,1}\left(\imath\omega\right)$, $\tilde \Xi_{out,2}\left(\imath\omega\right)$, $Z(\imath \omega)$  and  $\Xi(\imath \omega)$, respectively. 
Applying (\ref{eq:dynamics}), (\ref{eq:output}), we have
\begin{eqnarray}
\tilde \Xi_{out,1}^q(\imath \omega)+\tilde \Xi_{out,2}^q(\imath \omega) 
&=& \int_{-\infty}^{\infty} \xi_{out,1}^q(t)e^{-\imath\omega t} dt+\int_{-\infty}^{\infty} \xi_{out,2}^q(t)e^{-\imath\omega t}  dt \nonumber \\
&=& [1\ 0\ 1\ 0] \left(C Z\left(\imath\omega\right)+ D\Xi\left(\imath\omega\right)\right), \nonumber\\
\tilde \Xi_{out,1}^p(\imath \omega)-\tilde \Xi_{out,2}^p(\imath \omega)
&=& \int_{-\infty}^{\infty} \xi_{out,1}^p(t)e^{-\imath\omega t} dt-\int_{-\infty}^{\infty} \xi_{out,2}^p(t)e^{-\imath\omega t}  dt \nonumber \\
&=&[0\ 1\ 0\ {-}1]\left(C Z\left(\imath\omega\right)+ D\Xi\left(\imath\omega\right)\right).
\end{eqnarray}

The two-mode squeezing spectra $V_{+}(\imath \omega)$ and $V_{-}(\imath \omega)$ are real functions defined via the identities
\begin{eqnarray}
 \langle (\tilde \Xi_{out,1}^q(\imath \omega)+\tilde \Xi_{out,2}^q(\imath \omega))^* (\tilde \Xi_{out,1}^q(\imath \omega')+\tilde \Xi_{out,2}^q(\imath \omega')) \rangle &= & V_+(\imath \omega)\delta(\omega-\omega'), \nonumber \\
 \langle (\tilde \Xi_{out,1}^p(\imath \omega)-\tilde \Xi_{out,2}^p(\imath \omega))^* (\tilde \Xi_{out,1}^p(\imath \omega')-\tilde \Xi_{out,2}^p(\imath \omega')) \rangle  &=&  V_-(\imath \omega) \delta(\omega-\omega'),
\end{eqnarray}
where $\langle \cdot \rangle$ denotes quantum expectation. As described in \cite{Gough2010,Nurdin2012}, $V_+(\imath \omega)$ and $V_-(\imath \omega)$ are easily calculated by,
\begin{eqnarray}
V_+(\imath\omega)=& {\rm Tr}\left[H_1(\imath\omega)^* H_1(\imath\omega)\right], \label{eq:V_+}\\
V_-(\imath\omega)=& {\rm Tr}\left[H_2(\imath\omega)^* H_2(\imath\omega)\right], \label{eq:V_-}
\end{eqnarray}
where $H_1=[1\ 0\ 1\ 0]H$, $H_2=[0\ 1\ 0\ {-}1]H$ and $H$ is the transfer function
\begin{eqnarray}
H(\imath\omega)=C\left(\imath\omega I-A \right)^{-1}B+D. \label{eq:transfer-function}
\end{eqnarray}
The fields $\xi_{out,1}$ and  $\xi_{out,2}$ to be EPR-entangled at the frequency $\omega$ rad/s is \cite{Vitali2006},
\begin{eqnarray}
V(\imath\omega) = V_+(\imath\omega)+V_-(\imath\omega)< 4,  \label{eq:entanglement-criterion}
\end{eqnarray}
which indicates that the two-mode squeezing level is below the quantum shot-noise limit.

A perfect Einstein-Podolski-Rosen state is represented by an infinite bandwidth two-mode squeezing, that is $V(\imath\omega) = V_{\pm}(\imath\omega)= 0$ for all $\omega$. Of course, such an ideal EPR correlation cannot be achieved in reality as it would require an infinite amount of energy to produce. Thus, we aim to optimize EPR entanglement by making $V(\imath \omega)$ as small as possible over a wide frequency range \cite{Vitali2006}. 

Note that (\ref{eq:entanglement-criterion}) is a sufficient condition for EPR entanglement, with the two beams squeezed in amplitude and phase quadratures. However, in general, they may be squeezed in other quadratures. Hence, we give the following definition of EPR entanglement.
Let $\xi^{\psi_1}_{out,1}=e^{\imath \psi_1}\xi_{out,1}$, $\xi^{\psi_2}_{out,2}=e^{\imath \psi_2}\xi_{out,2}$ with $\psi_1, \psi_2 \in(-\pi,\pi]$ and denote the corresponding two-mode squeezing spectra between $\xi^{\psi_1}_{out,1}$ and $\xi^{\psi_2}_{out,2}$  as $V^{\psi_1, \psi_2}_\pm(\imath\omega,\psi_1, \psi_2)$.
\begin{definition}
Fields $\xi_{out,1}$ and  $\xi_{out,2}$ are EPR entangled at the frequency $\omega$ rad/s if $\exists ~\psi_1, \psi_2 \in(-\pi,\pi]$ such that
\begin{eqnarray}
 V^{\psi_1, \psi_2}_+(\imath\omega,\psi_1, \psi_2)+V^{\psi_1, \psi_2}_-(\imath\omega,\psi_1, \psi_2)< 4. \label{eq:entanglement-criterion-2}
\end{eqnarray}
Unless otherwise specified, throughout the paper, EPR entanglement refers to the case with $\psi_1=\psi_2=0$.
EPR entanglement is said to vanish at $\omega$ if there are no values of $\psi_1$ and $\psi_2$ satisfying the above criterion.
\end{definition}

\subsection{The nondegenerate optical parametric amplifier (NOPA)}
\label{sec:NOPA}
A NOPA ($G_i$) is an open linear quantum system containing a two-ended cavity with a pair of orthogonally polarized bosonic modes $a_i$ and $b_i$ which satisfy  $[a_i, a_j^*]=\delta_{ij}$, $[b_i, b_j^*]=\delta_{ij}$, $[a_i, b_j^*]=0$ and $[a_i, b_j]=0$. By assuming a strong undepleted coherent pump beam onto the $\chi^{(2)}$ nonlinear crystal inside the cavity, the pump can be treated as a classical field (hence, quantum vacuum fluctuations are ignored) and the interaction of the modes $a_i$ and $b_i$ with the pump is modelled by the two-mode squeezing Hamiltonian $H= \frac{\imath}{2} \epsilon\left( a_i^* b_i^*- a_ib_i\right)$, where $\epsilon$ is a real coefficient relating to the effective amplitude of the pump beam, for details see \cite{bGardiner2004, bCarmichael2008, bBachor2009}.

As shown in Fig.~\ref{fig:single-NOPA}, interactions between the NOPA and its environment are denoted by coupling operators as follows. Modes $a_i$ and $b_i$ are coupled to ingoing fields $\xi_{in,a,i}$ and $\xi_{in,b,i}$ via coupling operators $L_1=\sqrt{\gamma}a_i$ and $L_2=\sqrt{\gamma}b_i$, respectively. Unwanted amplification losses $\xi_{loss,a,i}$ and $\xi_{loss,b,i}$ impact the NOPA through operators $L_3=\sqrt{\kappa}a_i$ and $L_4=\sqrt{\kappa}b_i$, respectively. The constants $\gamma$ and $\kappa$ are damping rates of the outcoupling mirrors (from which the output fields emerge from the NOPA), and of the loss channels, respectively. 
Applying Section \ref{sec:linear_sys}, we have the dynamics of the NOPA as \cite{Ou1992, Nurdin2009, Collett1984, Gardiner1985}
\begin{eqnarray}
\dot{a_i}\left(t\right)&=&-\left(\frac{\gamma+\kappa}{2}\right)a_i\left(t\right)+\frac{\epsilon}{2}b_i^*\left(t\right)-\sqrt{\gamma}\xi_{in,a,i}\left(t\right)-\sqrt{\kappa}\xi_{loss,a,i}\left(t\right),\nonumber \\
\dot{b_i}\left(t\right)&=&-\left(\frac{\gamma+\kappa}{2}\right)b_i\left(t\right)+\frac{\epsilon}{2}a_i^*\left(t\right)-\sqrt{\gamma}\xi_{in,b,i}\left(t\right)-\sqrt{\kappa}\xi_{loss,b,i}\left(t\right), \label{eq:NOPA-dynamics1}
\end{eqnarray}
following the boundary conditions \cite{Ou1992, bGardiner2004}, we have outputs
\begin{eqnarray}
\xi_{out,a,i}\left(t\right)&=&\sqrt{\gamma}a_i\left(t\right)+\xi_{in,a,i}\left(t\right),\nonumber \\
\xi_{out,b,i}\left(t\right)&=&\sqrt{\gamma}b_i\left(t\right)+\xi_{in,b,i}\left(t\right). \label{eq:NOPA-dynamics2}
\end{eqnarray}

Define the following quadrature vectors of the NOPA, 
\begin{eqnarray}
z &=& [a^q_i, a^p_i, b^q_i, b^p_i]^T, \nonumber \\
\xi &=&[\xi^q_{in,a,i},\xi^p_{in,a,i},\xi^q_{in,b,i},\xi^p_{in,b,i}, \xi^q_{loss,a,i},\xi^p_{loss,a,i}, \xi^q_{loss,b,i},\xi^p_{loss,b,i}]^T,\nonumber\\
\xi_{out}&=&[\xi^q_{out,a,i},\xi^p_{out,a,i},\xi^q_{out,b,i},\xi^p_{out,b,i}]^T,
\end{eqnarray}
From (\ref{eq:dynamics}), (\ref{eq:output}) and (\ref{eq:transfer-function}), the transfer function of the NOPA is
\begin{eqnarray}
H_{N}=
\left[\begin{array}{cccccccc}
h_1 & 0 & h_2 & 0 & h_3 & 0 & h_4 & 0 \\
0 & h_1 & 0 & -h_2 & 0 & h_3 & 0 & -h_4 \\
h_2 & 0 & h_1 & 0 & h_4 & 0 & h_3 & 0 \\
0 & -h_2 & 0 & h_1 & 0 & -h_4 & 0 & h_3
\end{array} \right],\label{eq:NOPAtf}
\end{eqnarray}
where $h_j~ (j=1,2,3,4)$ are functions of the frequency $\omega$,
\begin{eqnarray}
h_1(\imath \omega) &=&\frac{\epsilon^2+\gamma^2-(\kappa+2\imath\omega)^2}{\epsilon^2-(\gamma+\kappa+2\imath\omega)^2},\nonumber\\
h_2(\imath \omega) &=&\frac{2\epsilon\gamma}{\epsilon^2-(\gamma+\kappa+2\imath\omega)^2},\nonumber\\
h_3 (\imath \omega) &=&\frac{2\sqrt{\gamma\kappa}(\gamma+\kappa+2\imath\omega)}{\epsilon^2-(\gamma+\kappa+2\imath\omega)^2},\nonumber\\
h_4 (\imath \omega) &=&\frac{2\epsilon\sqrt{\gamma\kappa}}{\epsilon^2-(\gamma+\kappa+2\imath\omega)^2}. \label{eq:h_freq_dependent}
\end{eqnarray}

As reported in \cite{Nurdin2009,Iida2012}, parameters of the NOPA are set as follows. We set the reference value for the transmissivity rate of the mirrors $\gamma_{r}=7.2\times10^7$ Hz. The pump amplitude $\epsilon$ is adjustable as $\epsilon=x\gamma_r$, where the variable $x$ satisfies $0<x\leq 1$. We fix the damping rate $\gamma=\gamma_r$ and set $\kappa=K\epsilon$ with $K=\frac{3\times 10^6}{\sqrt{2}\times 0.6 \times \gamma_r}$  based on the assumption that the value of $\kappa$ is proportional to the absolute value of $\epsilon$ and $\kappa= \frac{3 \times 10^6}{\sqrt{2}}$ when $\epsilon=0.6 \gamma_r$.   In this paper, we consider the NOPAs have infinite bandwidth case where we take the limit $\gamma_r\rightarrow\infty$ while keeping $\epsilon$ and $\gamma$ at a fixed ratio $\frac{\epsilon}{\gamma}=x$. In such a case, the transfer function of the NOPA in (\ref{eq:NOPAtf}) becomes a constant matrix with elements
\begin{eqnarray}
h_1&=&\frac{(1-K^2)x^2+1}{x^2-(1+Kx)^2},\nonumber\\
h_2&=&\frac{2x}{x^2-(1+Kx)^2},\nonumber\\
h_3&=&\frac{2\sqrt{Kx}(1+Kx)}{x^2-(1+Kx)^2},\nonumber\\
h_4&=&\frac{2x\sqrt{Kx}}{x^2-(1+Kx)^2}. \label{eq:h_coeffi_static}
\end{eqnarray}

It can be seen that for $\omega \ll \epsilon, \gamma, \kappa$, the constant scalar values of $h_1$ to $h_4$ given by (\ref{eq:h_coeffi_static}) in the infinite bandwidth limit approximates the frequency dependent values given in (\ref{eq:h_freq_dependent}) when the bandwidth is finite. Such an approximation is quite accurate for $\omega$ sufficiently small, away from $\epsilon, \gamma, \kappa$ (with no error at $\omega=0$). Since in practice the EPR entanglement will be in the low frequency region, entanglement in the idealised  infinite bandwidth scenario provides a good approximation for the entanglement that can be expected in the finite bandwidth case.

%%%%%%%%%%%%%%%%%%%%%%%%%%%%%%%%%%%%%%%%%%%%%%%%%%%%%%%%%%%%%%%%%%%%%%%%%%%%%%%%
%%%%%%%%%%%%%%%%%%%%%%%%%%%%%%%%%%%%%%%%%%%%%%%%%%%%%%%%%%%%%%%%%%%%%%%%%%%%%%%%
%%%%%%%%%%%%%%%%%%%%%%%%%%%%%%%%%%%%%%%%%%%%%%%%%%%%%%%%%%%%%%%%%%%%%%%%%%%%%%%%
%%%%%%%%%%%%%%%%%%%%%%%%%%%%%%%%%%%%%%%%%%%%%%%%%%%%%%%%%%%%%%%%%%%%%%%%%%%%%%%%
\section{The system model}
\label{sec:system-model}
Consider again the coherent feedback system shown in Fig.~\ref{fig:system}. The whole network consists of two NOPAs and a static passive linear subsystem. The subsystem has two inputs $\xi'_{out,a,1}$ and $\xi'_{out,b,2}$ connected to the outgoing fields $\xi_{out,a,1}$ of NOPA $G_1$ and $\xi_{out,b,2}$ of NOPA $G_2$, respectively. The two outputs $\xi'_{in,b,1}$ and $\xi'_{in,a,2}$ of the subsystem are connected to incoming signals $\xi_{in,b,1}$ of NOPA $G_1$ and $\xi_{in,a,2}$ of NOPA $G_2$, respectively. The incoming fields of the system $\xi_{in,a,1}$ and $\xi_{in,b,2}$ are in the vacuum state \cite{Nurdin2009} and the EPR entanglement of interest is generated between outgoing fields $\xi_{out,b,1}$ and $\xi_{out,a,2}$. 
The transfer function of the passive static subsystem is a $2 \times 2$ complex unitary matrix denoted by $\tilde S$, which satisfies \cite{Nurdin2009}
\begin{eqnarray}
\left[ \begin{array}{c}
\xi'_{in,b,1}\\ \xi'_{in,a,2} \end{array} \right] 
= \tilde S \left[ \begin{array}{c}
\xi'_{out,a,1}\\ \xi'_{out,b,2} \end{array} \right], \label{eq:S}
\end{eqnarray}
and 
\begin{equation}
\tilde S^* \tilde S = \tilde S \tilde S^*= I_2. \label{eq：complex-unitary}
\end{equation} 
Also, we shall denote the static passive matrix $\tilde{S}$ corresponding to the dual-NOPA coherent feedback network shown in Fig.~\ref{fig:dual-NOPA-cfb} as \cite{SN2015qip} 
\begin{eqnarray}
\tilde{S}_{cfb} &=& \left[\begin{array}{cc} 0 & 1\\ 1 & 0\end{array}\right].
\label{eq:Scfb} \end{eqnarray}

%%%%%%%%%%%%%%%%%%%%%%%%%%%%%%%%%%%%%%%%%%%%%%%%%%%%%%%%%%%%%%%%%%%%%%%%%%%%%%%%%%%%%%%%%%%%%%%%%%%%%%%%%%%%%%%%%%%%%%%%%%%%%%%%%%%%%%%%%%%%%%%%%%%%%%%%%%%%%%%%%%%%%%%%%%%%%%%%%%%%%%%%%%%%%%%%%%%%%%%%%%%%%%%%%%%%%%%%%%%%%%%%
\begin{figure}[htbp]
\begin{center}
\includegraphics[scale=0.6]{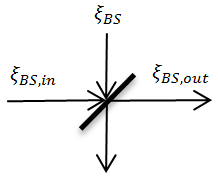}
\caption{Beamsplitter.}\label{fig:BS}
\end{center}
\end{figure}
%%%%%%%%%%%%%%%%%%%%%%%%%%%%%%%%%%%%%%%%%%%%%%%%%%%%%%%%%%%%%%%%%%%%%%%%%%%%%%%%%%%%%%%%%%%%%%%%%%%%%%%%%%%%%%%%%%%%%%%%%%%%%%%%%%%%%%%%%%%%%%%%%%%%%%%%%%%%%%%%%%%%%%%%%%%%%%%%%%%%%%%%%%%%%%%%%%%%%%%%%%%%%%%%%%%%%%%%%%%%%%%%

The NOPAs are placed at two distant communicating ends (Alice and Bob). The distance between the two ends is $d$ kilometres. Both NOPAs ($G_1$ and $G_2$) have  identical static transfer functions given by (\ref{eq:NOPAtf}) and (\ref{eq:h_coeffi_static}). Transmission loss in each path of the network is modelled by a beamsplitter with an unwanted incoming vacuum noise $\xi_{BS}$, as shown in Fig.~\ref{fig:BS}. 
The other input is connected to an outgoing field of the NOPAs or the subsystem. The outgoing signal $\xi_{BS,out}$ of the beamsplitter is the combination of the two incoming signals, satisfying $\xi_{BS,out}=\alpha\xi_{BS,in}+\beta\xi_{BS}$, where $\alpha$ is the transmission rate and $\beta$ is the reflection rate of the beamsplitter. $\alpha$ and $\beta$ are positive real parameters obeying $0\leq \alpha, \beta \leq 1$ and $\alpha^2+\beta^2=1$ \cite{bGerry2005}.
Based on the fact that transmission loss in optical fibre is about $0.2$ dB per kilometre at telecom wavelengths as reported in \cite{Jacobs2002}, the transmission rate of each beamsplitter in our system is $\alpha=10^{-0.005d}$.

Technically, transmission losses are accompanied by time delays in the transmission. However, here we neglect the time delays in transmission.  
Nonetheless, the formalism and Heisenberg picture analysis employed here can easily treat the presence of time delays in the linear quantum networks considered herein, see  \cite{SN2015qip, SN2015qic, Gough2010, Nurdin2012}.
In previous works on related studies \cite{SN2015qip, SN2015qic, Nurdin2012}, the effect of time delays has only been to narrow the bandwidth over which the EPR entanglement exists, without affecting the EPR entanglement that can be achieved in the low frequency region. 

Define the following vectors of quadratures
\begin{eqnarray}
z&=&[a^q_1, a^p_1, b^q_1, b^p_1,a^q_2, a^p_2, b^q_2, b^p_2]^T,\nonumber \\
\xi_{in}&=&[\xi^q_{in,a,1},\xi^p_{in,a,1},\xi^q_{in,b,2},\xi^p_{in,b,2}]^T, \nonumber\\
\xi_{loss}&=&[\xi^q_{loss,a,1},\xi^p_{loss,a,1},\xi^q_{loss,b,1},\xi^p_{loss,b,1},\xi^q_{loss,a,2},\xi^p_{loss,a,2},\xi^q_{loss,b,2},\xi^p_{loss,b,2}]^T, \nonumber\\
\xi_{BS}&=&[\xi^q_{BS,a,1},\xi^p_{BS,a,1},\xi^q_{BS,b,1},\xi^p_{BS,b,1},\xi^q_{BS,a,2},\xi^p_{BS,a,2},\xi^q_{BS,b,2},\xi^p_{BS,b,2}]^T, \nonumber\\
\xi &= &[\xi_{in}^T,\xi_{loss}^T,\xi_{BS}^T]^T, \nonumber\\
\xi_{out}&=&[\xi^q_{out,1},\xi^p_{out,1},\xi^q_{out,2},\xi^p_{out,2}]^T. \label{e1:sys_quadratures}
\end{eqnarray}
Define the real unitary matrix $S$ as the quadrature form of matrix $\tilde S$. Based on the definitions of the quadratures (\ref{eq: quadratures}), we obtain
\begin{eqnarray}
S=\frac{1}{2}\tilde{K}\tilde{S}\tilde{K}^* + \frac{1}{2}\tilde{K}^\# \tilde{S}^\#\tilde{K}^T,  \label{eq:relations-real-complex-matrix}
\end{eqnarray}
where 
\begin{eqnarray}
\tilde{K}= I_2 \otimes \left[ \begin{array}{c} 1\\ -\imath  \end{array} \right]. \label{eq:Ktilde}
\end{eqnarray}
Note the quadrature form $S$ is, by construction, a unitary symplectic matrix. That is, $S$ is unitary and symplectic, the latter meaning that $S^{\top} \left[ \begin{array}{cc} 0 & 1 \\ -1 & 0 \end{array} \right] S = \left[ \begin{array}{cc} 0 & 1 \\ -1 & 0 \end{array} \right]$. 

Using the static transfer function of a NOPA given by (\ref{eq:NOPAtf}) and (\ref{eq:h_coeffi_static}), and given the unitary matrix $\tilde S$ representing  the passive static subsystem, we obtain the static linear transformation $H(S)$ ($\xi_{out}=H(S)\xi$) of the dual-NOPA coherent feedback static system as a function of  $S$,
\begin{eqnarray}
H(S)= \tilde{H}_2+ h_1P[\alpha^2 S\tilde{H}_1~ H_{BS}],
\end{eqnarray}
where
\begin{eqnarray}
P &=& (I_4 - \alpha^2 S (I_2 \otimes \tilde{h}_2))^{-1}, \nonumber \\
H_{BS} &=& \beta \left[\begin{array}{ccc} O_{4 \times 2} & I_4 & O_{4 \times 2} \end{array} \right]+\alpha\beta S
\left[\begin{array}{cccccccc}
1& 0& 0& 0& 0& 0& 0& 0\\
0& 1& 0& 0& 0& 0& 0& 0\\
0& 0& 0& 0& 0& 0& 1& 0\\
0& 0& 0& 0& 0& 0& 0& 1
\end{array} \right],\nonumber\\
\tilde{H}_1 &=&
\left[\begin{array}{cccccc}
\tilde{h}_1 & O_2 & \tilde{h}_3 & \tilde{h}_4 & O_2 & O_2 \\
O_2 & \tilde{h}_1 & O_2 & O_2 & \tilde{h}_4 &\tilde{h}_3 
\end{array} \right],\nonumber\\
\tilde{H}_2 &=&
\left[\begin{array}{cc} \left[\begin{array}{cccccc}
\tilde{h}_2 & O_2 & \tilde{h}_4 & \tilde{h}_3 & O_2 & O_2 \\
O_2 & \tilde{h}_2 & O_2 & O_2 & \tilde{h}_3 &\tilde{h}_4 
\end{array} \right]& O_{4 \times 8} \end{array}\right],\nonumber\\
\tilde{h}_1&=&I_2 \otimes h_1, ~~~~~~~~~~~ \tilde{h}_3=I_2 \otimes h_3, \nonumber\\
\tilde{h}_2&=&\left[\begin{array}{cc} h_2 & 0\\ 0 & -h_2\end{array} \right], ~~ \tilde{h}_4=\left[\begin{array}{cc} h_4 & 0\\ 0 & -h_4\end{array} \right].\label{eq:Ht}
\end{eqnarray}
%%%%%%%%%%%%%%%%%%%%%%%%%%%%%%%%%%%%%%%%%%%%%%%%%%%%%%%%%%%%%%%%%%%%%%%%%%%%%%%%
%%%%%%%%%%%%%%%%%%%%%%%%%%%%%%%%%%%%%%%%%%%%%%%%%%%%%%%%%%%%%%%%%%%%%%%%%%%%%%%%
%%%%%%%%%%%%%%%%%%%%%%%%%%%%%%%%%%%%%%%%%%%%%%%%%%%%%%%%%%%%%%%%%%%%%%%%%%%%%%%%
%%%%%%%%%%%%%%%%%%%%%%%%%%%%%%%%%%%%%%%%%%%%%%%%%%%%%%%%%%%%%%%%%%%%%%%%%%%%%%%%
\section{Optimization of $\tilde S$}
\label{sec:optimization}
In this section, we aim to optimize the EPR entanglement generated in by the dual-NOPA coherent feedback  system of Fig.~\ref{fig:system},  in the infinite bandwidth limit, by  finding a complex unitary matrix at which the two-mode squeezing spectra of the two outgoing fields are locally minimized, with respect of  $\tilde{S}$.  Since the system is infinite bandwidth,  $V_\pm(\imath \omega)=V_\pm(0)$  for all $\omega$, thus we shall denote $V(\imath \omega)$ and $V_\pm(\imath \omega)$ simply as  $V$ and $V_\pm $, with no dependence on $\omega$. Based on (\ref{eq:V_+}), (\ref{eq:V_-}) and (\ref{eq:entanglement-criterion}), the sum of the two-mode squeezing spectra is
\begin{eqnarray}
V &=& V_++V_- \nonumber \\
&=&\operatorname{Tr}\left[H_1^*H_1+H_2^*H_2\right], \nonumber \\
&=&\operatorname{Tr}\left[H(S)^*M_{1,2} H(S)  \right] \label{eq:entanglement}
\end{eqnarray}
where
\begin{eqnarray}
M_{1,2} &=&  \left[\begin{array}{cccc} 1 & 0 & 1 & 0 \\ 0 & 1 & 0 & -1 \\1 & 0 & 1 & 0 \\0 & -1 & 0 & 1 \end{array}\right].  
\end{eqnarray}
As $V$ is a function of $\tilde S$ or $S$, we define $V(\tilde S)$ as the value of $V$ for a fixed value of $\tilde S$, and $V(S)$ as the value of $V$ for a fixed value of $S$. 

We aim to find a complex unitary matrix $\tilde S$ as a local minimizer of  the cost function $V(\tilde{S})$. The optimization problem with a unitary constraint can be solved by the method of modified steepest descent on a Stiefel manifold introduced in \cite{Manton2002}, which employs the first-order derivative of the cost function. The Stiefel manifold in our problem is the set $ St(2,2)=\left\lbrace\tilde{S} \in \mathbb{C}^{2 \times 2} : \tilde{S}^* \tilde{S} = I\right\rbrace$.  

Since $(I-Y)^{-1}=(I-Y)^{-1}(I+Y-Y)=I+(I-Y)^{-1}Y$ for any square matrix $Y$ such that $I-Y$ is invertible, we expand $H(S+\Delta S)$ as $H(S)+ H(\Delta S)+ H(\Delta S^2)+O(\Delta S^3)$, where $O(\Delta S^3)$ denotes terms that are products containing at least three $\Delta S$. $H(S)$, $H(\Delta S)$ and $H(\Delta S^2)$ are real matrices,
\begin{eqnarray}
H(\Delta S)&=&P \Delta S Q \nonumber\\
H(\Delta S^2)&=&\alpha^2 P \Delta S (I_2 \otimes \tilde{h}_2)P \Delta S Q, \label{H1H2}
\end{eqnarray}
where
\begin{eqnarray}
Q = \alpha^2 h_1 \left[\begin{array}{cc}\left(I_4 + \alpha^2 \left(I_2 \otimes \tilde{h}_2 \right) P S\right)\tilde{H}_1 & (I_2 \otimes \tilde{h}_2) P H_{BS} \end{array}\right].
\end{eqnarray}
Following (\ref{eq:entanglement}) and based on the facts that a matrix and its transpose have the same trace, we have
\begin{eqnarray}
 V(S +\Delta S) &=&\operatorname{Tr} \left[H(S+\Delta S)^* M_{1,2} H(S+\Delta S) \right] \nonumber\\
&=& V(S) + \operatorname{Tr}[H(\Delta S)^*M_{1,2} H(S) + H(S)^*M_{1,2} H(\Delta S) +H(\Delta S^2)^*M_{1,2} H(S) \nonumber\\
&& \quad + H(S)^*M_{1,2} H(\Delta S^2)+ H(\Delta S)^*M_{1,2} H(\Delta S)]+ O(\lVert\Delta S\rVert^3) \nonumber\\
&=& V(S) + 2\operatorname{Tr}[M H(\Delta S)]+ 2\operatorname{Tr}[M H(\Delta S^2)]\nonumber\\
&& \quad + \operatorname{Tr}[H(\Delta S)^*M_{1,2} H(\Delta S)] +O(\lVert\Delta S\rVert^3),
\end{eqnarray}
where $M=H(S)^*M_{1,2}$ and $O(\lVert\Delta S\rVert^3)$ denotes that the function $O(\lVert\Delta S\rVert^3)$ satisfies $\frac{O(\lVert\Delta S\rVert^3)}{\lVert\Delta S\rVert^3} \leq c$ for some positive constant $c$ for all $\lVert\Delta S\rVert>0$ sufficiently small.
Furthermore, based on (\ref{eq:relations-real-complex-matrix}), we obtain that
\begin{eqnarray}
 V( \tilde{S} + \Delta \tilde{S}) 
=V(\tilde{S})+ \operatorname{Re}\operatorname{Tr}[\Delta\tilde{S}^*D_{\tilde{S}}] +\frac{1}{2}\left[\begin{array}{c}
\operatorname{vec}(\Delta\tilde{S})\\
\operatorname{vec}(\Delta\tilde{S}^\#)
\end{array} \right]^* X \left[\begin{array}{c}
\operatorname{vec}(\Delta\tilde{S})\\
\operatorname{vec}(\Delta\tilde{S}^\#)
\end{array} \right]+ O(\lVert\Delta\tilde S\rVert^3), \label{eq:expansion_V}
\end{eqnarray}
where 
\begin{eqnarray}
D_{\tilde{S}} &=& 2\tilde{K}^* (QMP)^T \tilde{K} ,\nonumber \\
X &=& \frac{1}{4} \left[\begin{array}{cc}(\tilde{K}^\# \otimes \tilde{K}) &(\tilde{K}^\# \otimes \tilde{K})^\# \end{array}\right]^* h \left[\begin{array}{cc}(\tilde{K}^\# \otimes \tilde{K}) &(\tilde{K}^\# \otimes \tilde{K})^\# \end{array}\right],\nonumber\\
h &=&4\alpha^2 L^T(QMP)^T \otimes ((I_2 \otimes \tilde{h}_2)P) +2(QQ^T) \otimes (P^TM_{1,2}P), \label{eq:DsX}
\end{eqnarray}
\begin{eqnarray}
L=\left[\begin{array}{cccccccccccccccc} 1&0&0&0&0&0&0&0&0&0&0&0&0&0&0&0\\
0&0&0&0&1&0&0&0&0&0&0&0&0&0&0&0\\
0&0&0&0&0&0&0&0&1&0&0&0&0&0&0&0\\
0&0&0&0&0&0&0&0&0&0&0&0&1&0&0&0\\
0&1&0&0&0&0&0&0&0&0&0&0&0&0&0&0\\
0&0&0&0&0&1&0&0&0&0&0&0&0&0&0&0\\
0&0&0&0&0&0&0&0&0&1&0&0&0&0&0&0\\
0&0&0&0&0&0&0&0&0&0&0&0&0&1&0&0\\
0&0&1&0&0&0&0&0&0&0&0&0&0&0&0&0\\
0&0&0&0&0&0&1&0&0&0&0&0&0&0&0&0\\
0&0&0&0&0&0&0&0&0&0&1&0&0&0&0&0\\
0&0&0&0&0&0&0&0&0&0&0&0&0&0&1&0\\
0&0&0&1&0&0&0&0&0&0&0&0&0&0&0&0\\
0&0&0&0&0&0&0&1&0&0&0&0&0&0&0&0\\
0&0&0&0&0&0&0&0&0&0&0&1&0&0&0&0\\
0&0&0&0&0&0&0&0&0&0&0&0&0&0&0&1
\end{array}\right].\nonumber
\end{eqnarray} 
$D_{\tilde{S}}$ is the directional derivative of $V(\tilde{S})$ at $\tilde{S}$ in the direction $\Delta \tilde S$ \cite{Manton2002}.

\begin{theorem}
\label{th:critical_point} The matrix $\tilde{S}_{cfb}$ corresponding to the dual-NOPA coherent feedback system given by (\ref{eq:Scfb}) is a critical point of the function $V(\tilde{S})$.
\end{theorem}
\begin{proof}
According to \cite{Manton2002}, we have a one-to-one corresponding cost function $g_{\tilde S}(\Delta\tilde{S})$ on the tangent space to the Stiefel manifold at the point $\tilde S$, with $\Delta \tilde S$ a vector on this tangent space, defined by $g_{\tilde S}(\Delta\tilde{S})=V(\pi(\tilde{S}+\Delta\tilde{S}))$, where $\pi(\cdot)$ is the projection operator onto the manifold. The descent direction $Z_d$ in the tangent space at $\tilde{S}$ is 
\begin{eqnarray}
Z_d=\tilde{S}D_{\tilde{S}}^*\tilde{S}-D_{\tilde{S}}.
\end{eqnarray}
Based on (\ref{eq:DsX}), when $\tilde{S}=\tilde{S}_{cfb}$,  $D_{\tilde{S}}$ becomes
\begin{eqnarray}
D_{\tilde{S}_{cfb}}&=& d_{\tilde{S}_{cfb}} \left[\begin{array}{cc} 0 & 1\\ 1 & 0\end{array}\right],
\label{eq:DScfb} \end{eqnarray}
where $d_{\tilde{S}_{cfb}}$ is a real coefficient
\begin{eqnarray}
d_{\tilde{S}_{cfb}}&=&\frac{1}{(1 + 2 \alpha^2 x + 2K x - x^2 +K^2 x^2)^3}(4 \alpha^2 (-1 \nonumber\\
&&+ (-1 + K^2) x^2) (4 x (1 + 2 K x + x^2 + K^2 x^2)  \nonumber\\
&&+ 2 \alpha^4 x (-1 + (-1 + K^2) x^2) + \alpha^2 (-1 - 2 K x \nonumber\\
&&+ 6 K x^3 + 2 K^3 x^3 + x^4 - 2 K^2 x^4 + K^4 x^4))).
\end{eqnarray}
Thus, the descent direction $Z_d$ is
\begin{eqnarray}
Z_{d,cfb}=\tilde{S}_{cfb}D_{\tilde{S}_{cfb}}^*\tilde{S}_{cfb}-D_{\tilde{S}_{cfb}}=O_2.
\end{eqnarray}
Thus, the gradient of the  function $g_{\tilde S}(\Delta\tilde{S})$ at $\Delta \tilde S=0$ along the tangent space at $\tilde S_{cfb}$ is $\operatorname{grad}\left(g_{\tilde S}(0)\right)=\tilde{S}_{cfb}D_{\tilde{S}_{cfb}}^*\tilde{S}_{cfb}-D_{\tilde{S}_{cfb}}=O_2$ (see \cite{Manton2002}[Eq. (27)]), which establishes that $\tilde{S}_{cfb}$ is a critical point.
\end{proof}

Now we check the Hessian matrix of the function $g_{\tilde S}(\Delta\tilde{S})$.  Based on Proposition $12$ in \cite{Manton2002} and (\ref{eq:expansion_V}), we have following the second order expansion along any direction $\Delta \tilde S$ on the tangent space at $\tilde S$,
\begin{eqnarray}
g_{\tilde S}(\Delta\tilde{S})&=& V(\pi(\tilde{S} + \Delta \tilde S)) \nonumber\\
&=&V(\tilde{S})+ \operatorname{Re}\operatorname{Tr}[\Delta\tilde{S}^*D_{\tilde{S}}]+\frac{1}{2}\left[\begin{array}{c}
\operatorname{vec}(\Delta\tilde{S})\\
\operatorname{vec}(\Delta\tilde{S}^\#)
\end{array} \right]^* {\rm Hess}({\tilde S})\left[\begin{array}{c}
\operatorname{vec}(\Delta\tilde{S})\\
\operatorname{vec}(\Delta\tilde{S}^\#)
\end{array} \right]\nonumber\\
&&+ O(\lVert\Delta\tilde S\rVert^3), \label{eq:expansion_V2}
\end{eqnarray}
where 
\begin{eqnarray}
{\rm Hess}({\tilde S})=X-\frac{1}{2} \left[\begin{array}{cc}
(\tilde{S}^*D_{\tilde{S}})^T\otimes I_2 &O_4\\
O_4 & ((\tilde{S}^*D_{\tilde{S}})^T\otimes I_2 )^\#
\end{array} \right]\label{eq:Hess}
\end{eqnarray}
denotes the Hessian matrix of $g_{\tilde S}(\Delta\tilde{S})$. Firstly, we consider the system in an ideal case, where there are no losses ($\kappa=0$ and $\alpha=1$). As reported in \cite{SN2015qip}, in this lossless scenario the range of $x$ over which the dual-NOPA coherent feedback system is stable in the finite bandwidth case is $x \in [0, \sqrt{2}-1)$, independently of the actual bandwidth of the NOPAs. Thus, it is natural to also take this as the range of admissible values for $x$ in the infinite bandwidth limit of this paper. By checking eigenvalues of the Hessian matrix, we have the following theorem.
\begin{theorem}
\label{th:ideal} In the absence of transmission and amplification losses, $\tilde{S}_{cfb}$ is a local minimizer of the function $V(\tilde{S})$ when $x \in (\sqrt{5}-2, \sqrt{2}-1)$. 
%   a local maximizer when $x \in(0, 0.127375]$ and a saddle point if $x \in(0.127375,  0.355182)$.
\end{theorem}
\begin{proof}
Let $\alpha=1$ and $\kappa=0$. With the help of  Mathematica, the eigenvalues of ${\rm Hess}(\tilde S)$ at $\tilde{S}=\tilde{S_{cfb}}$ can be found to be
\begin{eqnarray}
e_1&=&\frac{8x(1-x^2)(1+x^2)^2}{(1+2x-x^2)^4}, \nonumber\\
e_2&=&\frac{8x(1-x^2)(1+x^2)^2}{(1-6x^2+x^4)^2},\nonumber\\
e_3&=&\frac{8x(1+x^2)^2(-1+4x+x^2)}{(1+2x-x^2)^4},\nonumber\\
e_4&=&\frac{8x(1+x^2)^2(3-6x+2x^2+6x^3+3x^4)}{(1+2x-x^2)^3(1+2x+x^2)^2}.
\label{eq:eigenvalues}
\end{eqnarray}
As  $x\in(0, \sqrt{2}-1)$,  $e_1, e_2$ and $e_4$ have positive values, while $e_3>0$ when $-1+4x+x^2>0$, that is, $x>\sqrt{5}-2$. Therefore, for   $x\in(\sqrt{5}-2, \sqrt{2}-1)$, the Hessian matrix ${\rm Hess}(\tilde S_{cfb})$ is positive definite, which establishes that  $\tilde{S}_{cfb}$ is a local minimizer for these values of $x$.
\end{proof}

Table \ref{tb: transmission} and Table \ref{tb: amplification} illustrate the effect of transmission and amplification losses on the range $(x_{lm},~\sqrt{2}-1)$ of over which $\tilde{S}_{cfb}$ is a local minimizer. We see that as either transmission losses or amplification losses increase, the range of values of $x$ over which the dual-NOPA coherent feedback network is optimal become wider.
%%%%%%%%%%%%%%%%%%%%%%%%%%%%%%%%%%%%%%%%%%%%%%%%%%%%%%%%%%%%%%%%%%%%%%%%%%%%%%%%%%%%%%%%%%%%%%%%%%%%%%%%%%%%%%%%%%%%%%%%%%%%%%%%%%%%%%%%%%%%%%%%%%%%%%%%%%%%%%%%%%%%%%%%%%%%%%%%%%%%%%%%%%%%%%%%%%%%%%%%%%%%%%%%%%%%%%%%%%%%%%%%
\begin{table}[htbp]
\centering
\caption{Influence of transmission losses on the range $(x_{lm},~\sqrt{2}-1)$ of over which $\tilde{S}_{cfb}$ is a local minimizer with $\kappa=0$ and $\alpha=10^{-0.005d}$}\label{tb: transmission}
\begin{tabular}{|c|c|c|}
\hline
$d$ & $x_{lm}$ \\
\hline
$0$ &  $0.236068$\\
\hline
$1$ &  $0.212692$\\
\hline
$5$ &  $0.134477$\\
\hline
\end{tabular}
\end{table}
%%%%%%%%%%%%%%%%%%%%%%%%%%%%%%%%%%%%%%%%%%%%%%%%%%%%%%%%%%%%%%%%%%%%%%%%%%%%%%%%
%%%%%%%%%%%%%%%%%%%%%%%%%%%%%%%%%%%%%%%%%%%%%%%%%%%%%%%%%%%%%%%%%%%%%%%%%%%%%%%%%%%%%%%%%%%%%%%%%%%%%%%%%%%%%%%%%%%%%%%%%%%%%%%%%%%%%%%%%%%%%%%%%%%%%%%%%%%%%%%%%%%%%%%%%%%%%%%%%%%%%%%%%%%%%%%%%%%%%%%%%%%%%%%%%%%%%%%%%%%%%%%%
\begin{table}[htbp]
\centering
\caption{Influence of amplification losses on the range $(x_{lm},~\sqrt{2}-1)$ of over which $\tilde{S}_{cfb}$ is a local minimizer with $d=1$ and $\alpha=10^{-0.005d}$}\label{tb: amplification}
\begin{tabular}{|c|c|c|}
\hline
$\kappa$ & $x_{lm}$ \\
\hline
$0$ & $0.212692$\\
\hline
$0.1\frac{3 \times 10^6}{\sqrt{2} \times 0.6}x$ &$0.211836$ \\
\hline
$0.2\frac{3 \times 10^6}{\sqrt{2} \times 0.6}x$ & $0.210989$\\
\hline
$0.5\frac{3 \times 10^6}{\sqrt{2} \times 0.6}x$ & $0.208503$\\
\hline
$\frac{3 \times 10^6}{\sqrt{2} \times 0.6}x$ &  $0.204528$\\
\hline
\end{tabular}
\end{table}
%%%%%%%%%%%%%%%%%%%%%%%%%%%%%%%%%%%%%%%%%%%%%%%%%%%%%
%%%%%%%%%%%%%%%%%%%%%%%%%%%%%%%%%%%%%%%%%%%%%%%

%%%%%%%%%%%%%%%%%%%%%%%%%%%%%%%%%%%%%%%%%%%%%%%%%%%%%%%%%%%%%%%%%%%%%%%%%%%%%%%%
%%%%%%%%%%%%%%%%%%%%%%%%%%%%%%%%%%%%%%%%%%%%%%%%%%%%%%%%%%%%%%%%%%%%%%%%%%%%%%%%
%%%%%%%%%%%%%%%%%%%%%%%%%%%%%%%%%%%%%%%%%%%%%%%%%%%%%%%%%%%%%%%%%%%%%%%%%%%%%%%%
%%%%%%%%%%%%%%%%%%%%%%%%%%%%%%%%%%%%%%%%%%%%%%%%%%%%%%%%%%%%%%%%%%%%%%%%%%%%%%%%

\section{Conclusion}
\label{sec:conclusion}
This paper has studied the optimization of EPR entanglement of a static linear quantum system that is composed of a static linear passive optical network in a certain coherent feedback configuration with two NOPAs in the infinite bandwidth limit. We reformulate the optimization of the EPR entanglement to the problem of finding a $2 \times 2$ complex unitary matrix at which a cost function $V(\tilde{S})$ is locally minimized, with respect of $\tilde{S}$.  By employing the modified steepest descent on Stiefel manifold method, we have found the unitary matrix $\tilde S_{cfb}$ corresponding to the coherent feedback system shown as Fig.~\ref{fig:dual-NOPA-cfb} as a critical point of $V(\tilde{S})$.  When losses are neglected, the coherent feedback system is a local  minimizer when $x\in (\sqrt{5}-2, \sqrt{2}-1)$. When transmission and amplification losses increase, the range of values of $x$ over which the coherent feedback system is a local minimizer of $V(\tilde s)$ is enlarged. 
In addition, one may wonder if there exists other local minimizers at which the system generates better EPR entanglement. Hence future work can consider further developing the static passive optical network to search for another local optimizer that may yield better EPR entanglement than the system studied in \cite{SN2015qip} as shown in Fig.~\ref{fig:dual-NOPA-cfb}.

%%%%%%%%%%%%%%%%%%%%%%%%%%%%%%%%%%%%%%%%%%%%%%%%%%%%%%%%%%%%%%%%%%%%%%%%%%%%%%%%

\end{document}